\documentclass[journal]{IEEEtran}

\usepackage{enumerate}

\usepackage{ifthen}
\usepackage{amsmath}

\newboolean{showcomments}
\setboolean{showcomments}{true}
\newcommand{\masoud}[1]{  \ifthenelse{\boolean{showcomments}}
{ \textcolor{red}{(Masoud says:  #1)}} {}  }
\newcommand{\chris}[1]{\ifthenelse{\boolean{showcomments}}
{ \textcolor{red}{(Chris says: #1)} } {} }
\newcommand{\slow}[1]{\ifthenelse{\boolean{showcomments}}
{ \textcolor{red}{(Steven says:  #1)}}{}}
\newcommand{\mani}[1]{\ifthenelse{\boolean{showcomments}}
{ \textcolor{red}{(Mani says:  #1)}}{}}

\usepackage{graphicx,amssymb,amstext,amsmath}
\usepackage{xcolor}
\usepackage{multirow}
\usepackage{lscape}
\usepackage{mdwlist}
\usepackage{url}
\usepackage{dcolumn}
\usepackage{amsfonts}
\usepackage{mathrsfs}
\usepackage{latexsym}
\usepackage{graphicx}
\usepackage{subfigure}

\usepackage[latin1]{inputenc}
\usepackage{tikz}
\usepackage{pseudocode}
 \usepackage{fancybox}
\usetikzlibrary{shapes,arrows}

\tikzstyle{decision} = [diamond, draw,
    text width=4.5em, text badly centered, node distance=3cm, inner sep=0pt]
\tikzstyle{block} = [rectangle, draw,
    text width=5em, text centered, rounded corners, minimum height=4em]
\tikzstyle{line} = [draw, -latex']
\tikzstyle{cloud} = [draw, ellipse, node distance=3cm,
    minimum height=2em]

\newtheorem{theorem}{Theorem}

\def\ba{\begin{array}}
\def\ea{\end{array}}
\newcommand{\beq}{\begin{equation}}
\newcommand{\eeq}{\end{equation}}
\newcommand{\bq}{\begin{eqnarray}}
\newcommand{\eq}{\end{eqnarray}}
\newcommand{\bqn}{\begin{eqnarray*}}
\newcommand{\eqn}{\end{eqnarray*}}
\newcommand{\bee}{\begin{enumerate}}
\newcommand{\eee}{\end{enumerate}}
\newcommand{\bi}{\begin{itemize}}
\newcommand{\ei}{\end{itemize}}

\newcommand{\btab}{\begin{tabular}}
\newcommand{\etab}{\end{tabular}}

\begin{document}


\title{Optimal Inverter VAR Control in Distribution Systems with
High PV Penetration}


\author{\IEEEauthorblockN{Masoud Farivar\IEEEauthorrefmark{2}\IEEEauthorrefmark{1},
Russell Neal\IEEEauthorrefmark{2},
Christopher Clarke \IEEEauthorrefmark{2},
Steven Low \IEEEauthorrefmark{1}\\
\IEEEauthorblockA{\IEEEauthorrefmark{2}Southern California Edison, Rosemead, CA, USA}\\
\IEEEauthorblockA{\IEEEauthorrefmark{1}Department of Electrical Engineering,
Caltech, CA, USA}} }

\maketitle

\begin{abstract}
The intent of the study detailed in this paper is to demonstrate the benefits of inverter
var control on a fast timescale to mitigate rapid and large voltage fluctuations due to the
high penetration of photovoltaic generation and the resulting reverse power flow.  Our approach is to formulate the volt/var control as a radial optimal power flow (OPF)
problem to minimize line losses and energy consumption, subject
to constraints on voltage magnitudes.
An efficient solution to the radial OPF problem is presented and used to study the structure of
optimal inverter var injection and the net benefits, taking into account the additional cost of inverter losses when operating at non-unity power factor. This paper will illustrate how,
depending on the circuit topology and its loading condition, the inverter's optimal reactive
power injection is not necessarily monotone with respect to their real power output. The results
are demonstrated on a distribution feeder on the Southern California Edison system that has a
very light load and a 5 MW photovoltaic (PV) system installed away from the substation.

\end{abstract}

\vspace{2mm}
\begin{keywords}
Distribution systems, volt/var control, DC/AC inverter, optimal power flow, photovoltaics (PV) generation
\end{keywords}

\section{Introduction}

Sustainability of electric power systems requires development and massive integration of renewable energy sources. California has embarked on several initiatives to reach its ambitious goals in increasing the share of renewable energy
in its total energy mix. One in particular, the California Solar Initiative (CSI), is aimed at realizing 3,000 MW of new solar generation by 2016. As a result, distribution planners of the state's electric utilities are facing a rapidly increasing number of
integration request for small size residential PV as well as large-scale commercial PV systems. Solar energy is highly intermittent and this introduces several challenges to existing utility operation and control methods. One major hurdle
will be the volt/var control (VVC) in distribution circuits, which necessitates more advanced designs for much faster monitoring and control systems.

Regulating the voltage in distribution circuits within the acceptable range specified by the American National Standard Institute (ANSI) Standard C84.1 for power quality is an essential responsibility for utility companies. Traditionally voltage profile and reactive power flow in distribution feeders has been locally controlled using switched devices such as shunt capacitor banks, on-load tap changers (OLTCs) and voltage regulators. These devices are expected to switch only a few times a day to accommodate relatively slow variations in load, but this may not be sufficient for coping with the more rapid fluctuations associated with renewable generation. In addition to their active power, many
inverters have the capability to inject or absorb reactive power. However, according to IEEE 1547, the existing standard for integration of distributed energy resources (DERs), inverters should not actively participate in voltage/var regulation. This limitation may have been appropriate for low levels of penetration, but as utilities move towards higher penetration of renewables, it becomes necessary to exploit the advanced control capability of inverter interfaces to the grid.

We augment the traditional volt/var control through
switched controllers on a slow timescale with inverter control
on a fast timescale. The approach will be to cast the VVC
problem as optimal power flow (OPF) problem and solve it using conic relaxation of DistFlow representation of radial power flow equations. Constrained by voltage and reactive power flow limits, our objective function includes not only minimizing line losses, but also minimizing energy consumption through Conservation Voltage Reduction (CVR) and inverter losses.

The rest of this paper is organized as follows.
Section II discusses the potential challenges in voltage regulation
of distribution circuits  due to the high penetration of intermittent energy resources and the potential benefits of using DC/AC inverters to mitigate this problem and save energy. Section III mathematically formulates the objectives and constraints of the inverter var control problem and provides a Second Order Cone Program (SOCP) solution.  In Section IV,
we evaluate the proposed inverter var control using one of
SCE's distribution feeder that has a light load and a large 5MW
PV system installed far from the substation. We explain the
structure of the optimal inverter var injection as solar output
and as load vary, and demonstrate the improvement in voltage regulation and
efficiency under the optimal var control. We conclude in
Section V.

\section{Motivation for inverter var control}
The connection of large amounts of solar power at the distribution level presents a number of technical issues to address, such as load following, resource adequacy for contingencies, stability (given loss of system inertia), low voltage ride-through, line capacity, short circuit contribution, protection impacts of bi-directional power flow, distribution system planning and operation (load rolling), and voltage control.

When a large solar generator is interconnected to a distribution circuit, the real power it injects tends to cause a local voltage rise due primarily to the substation bus voltage and the resistance of the circuit back to the substation. In some cases this rise is large, and due to variable output may cause adverse voltage fluctuations for other connected customers. These fluctuations may also cause utility voltage regulating elements such as line regulators and capacitors to operate too frequently. The higher voltage will also work against Conservation Voltage Reduction (CVR) strategies. Since 1976, California has required utilities to provide voltage in the lower half of the ANSI C84.1 range (114 to 120 volts) for CVR purposes. Significant energy savings are attainable by better regulation of customer supply voltage, and for many years Southern California Edison (SCE) has explored closed loop capacitor control methods for this purpose \cite{R10}.

Southern California Edison operates a number of distribution circuits with high levels of photovoltaic generator penetration. Under the Solar Photovoltaic Program (SPVP) 500 MW of warehouse rooftop and ground mounted commercial generation in the 1 to 10 MW range are being deployed. This paper uses a time-series model of one such rural 12 kV feeder with a 5 MW generator near the end of the circuit to explore the capability of the generator inverter's reactive power to control voltage. This circuit is particularly interesting because of its low loading, high generation, long distance of the generator from the substation, and large reverse power flow.
Figure \ref{reverse}, shows the feeder current data measured at the substation, taken from the SCADA system of SCE. In this plot, a positive current shows a reverse power flow back to the substation. As you can see, this feeder can easily have a peak reverse flow
of more than 3MW. Four typical days in November 2011 were chosen to represent the four classes of solar radiation behavior introduced in \cite{Ted09}; i.e., clear, cloudy, intermittent clear, and intermittent cloudy days. The later type is the worst case in terms of voltage fluctuations.

Among the means being considered to mitigate these voltage effects is the use of the inverter's inherent reactive power capability to offset its real power effect. An inverter will often have a kVA size on the order of 110\% of its maximum kW output. This leaves 46\% of its capacity for reactive power even at full real power output. Since voltage rise as a function of kvar on a typical distribution circuit is 2 to 3 times that of a function of kW significant voltage regulation is possible.

The aim of this study it to reveal how the reactive power capability of the solar plant inverter can be used to regulate voltage and determine a voltage regulation operation that minimizes power losses by considering (1) CVR effects, (2) line losses, and (3) losses in the inverter due to reactive power flow. Maintaining customer voltage within ANSI C84.1 limits is treated as a constraint. Maintaining circuit power factor at the substation is not treated as a constraint.

 \begin{figure}[t]
\centering
\includegraphics[scale=0.48]{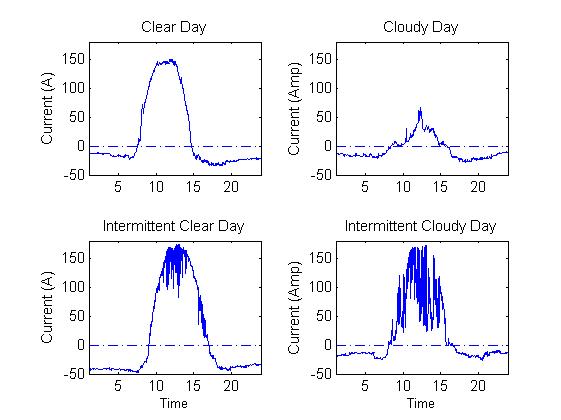}
\caption{\label{SP} Line current measurement at the substation for one of SCE's lightly loaded 12KV feeders with 5MW of PV installed almost at the end of the line. A positive current represents reverse power flow into the substation and a negative current shows real power flowing into the feeder. The plots are from SCADA data of SCE for 4 days in Nov 2011.}\label{reverse}
\end{figure}


\section{Problem formulation and solution}

We will use the balanced radial power flow equations first introduced
in \cite{B89}, called DistFlow equations.  Let $G (N, E)$  be a graph representing a radial distribution circuit.
Each node in $N$ is a bus and each link in $E$ is a line.
We index the nodes by $i = 1, \dots, |N|$. Let $V_i$, $i \geq 0$, denote the complex voltage at node $i$ and $I_{ij}, P_{ij}, Q_{ij}$ be the complex current, active and reactive powers flowing from node $i$ to $j$, respectively.

For node $i$, let $p_i^c$ and $q_i^c$  be the real and reactive power demand,
respectively.
Also let $p_i^g$ and $q_i^g$ be the real and reactive power generation
from the PVs and inverters, respectively, at node $i$.
If there is no load at node $i$, we assume $p_i^c = q_i^c = 0$;
similarly, if there is no PV at node $i$, we set $p_i^g = q_i^g = 0$.
Here, $p_i^c, q_i^c$, and $p_i^g$ are assumed to be given quantities,
whereas the reactive power generations $q_i^g$ are the control variables.

Finally some nodes have  shunt capacitors that are reconfigured on a slow timescale to provide reactive power.
let us define $q_i^{sc}$ to be the rating of the capacitor at node $i$. Then the reactive power generated by the shunt capacitor at node $i$ will be $q_i^{sc} |V_i|^2$ where $q_i^{sc}$ is the reactive power
generated by the shunt capacitor at node $i$ when $|V_i|=1$. If node $i$ has no shunt capacitor or has a shunt capacitor that is turned off in the current state of the slow timescale control, we just simply let $q_i^{sc} = 0$.

Then, from \cite{B89}, these
variables satisfy the following recursion (Dist-Flow equations): for each link
$(i, j)$ in the distribution circuit,
\begin{flushleft}
\begin{eqnarray}
P_{ij}  & = & \sum_{k:(j,k)\in E} {P_{jk}}+r_{ij} \frac{P_{ij}^2+Q_{ij}^2}{|V_i|^2} + p_j^c-p_j^g
\nonumber \\
& &
 \label{Pbalance}\\
Q_{ij}  & = & \sum_{k:(j,k)\in E} {Q_{jk}}+x_{ij} \frac{P_{ij}^2+Q_{ij}^2}{|V_i|^2}
  + q_j^c -q_j^g - q_j^{sc} |V_j|^2
\nonumber \\
& &
\label{Qbalance}\\
|V_j|^2
 & = &|V_i|^2 -2(r_{ij}P_{ij} +x_{ij}Q_{ij})
 + (r_{ij}^2+x_{ij}^2) \frac{P_{ij}^2+Q_{ij}^2}{|V_i|^2}
\nonumber \\
& &
\label{Vdrop}
\end{eqnarray}
\end{flushleft}

In this model, the current magnitude at each link can be determined using the following equation:
\begin{equation}
|I_{ij}|^2 = \frac{P_{ij}^2+Q_{ij}^2}{|V_i|^2}
\label{update}
\end{equation}

\subsection{Constraints}
Aside from the power flow equations (\ref{Pbalance})--(\ref{update}), we will now consider the constraints on voltages and reactive power flows:

\noindent 1) The primary purpose of VVC on distribution circuits is to maintain
voltages in an acceptable range.  This is formulated as constraints on the voltage variables $V_i$,
for all $i\geq 0$ :
\begin{eqnarray}
\underline{V}_i & \leq \ \ |V_i| \ \ \leq  & \overline{V}_i
 \end{eqnarray}

\noindent 2) The magnitude of the reactive power $q_i^g$
generated at an inverter  is upper bounded by a quantity that depends
on the real power generated at node $i$:
\begin{eqnarray}\label{inverter-model}
|q_i^{g}
|  & \leq & \overline{q}_i^g
\end{eqnarray}
where $\overline{q}_i^g := \sqrt{\overline{s}_i^2-\left(p_i^{g}
\right)^2}$, and $\overline{s}_i$ is the nameplate capacity of the inverter at node $i$. This
bound is assumed to be known at each time.
\subsection{Objective Function}
As mentioned earlier, to minimize overall power consumption, we consider line losses, CVR, and inverter losses in our objective function.
\subsubsection{Line losses} This can be simply formulated as
\begin{equation}\sum_{(i,j)\in E}{r_{ij} |I_{ij}|^2} \label{line-losses}
\end{equation}
\subsubsection{CVR}
As discussed in \cite{M11}, considering an exponential power consumption
model, maxmizing CVR savings
is equivalent to minimizing
\begin{equation}\sum{ \alpha_i |V_i|^2} \label{CVR}
\end{equation}
over all
voltage dependent loads where $\alpha_i = (n_i/2)p_i^c$ and $0 \leq n_i \leq 2$ is the exponent factor of consumption model of the load at node $i$. Three special cases are of particular interest: $n_i=0$ for constant power loads,
$n_i=1$ for constant current loads, and $n_i=2$ for constant impedance loads.
\subsubsection{Inverter losses}
DC/AC inverters are not perfect, they have losses. The real power loss in an inverter can be approximated by a quadratic function of its apparent power \cite{Braun07}:
 \begin{eqnarray}
P_{\textrm{loss}}(s) & = &  c_s + c_v s + c_r s^2\\
 & = & c_s + c_v \sqrt{p^2+q^2} + c_r (p^2+q^2) \label{inverter-loss}
 \end{eqnarray}
where $c_s$ models the inverter's standby losses, $c_v$ is the voltage dependent
losses over the power electronic components which is proportional to its current,
$I$, $c_r$ is the ohmic losses proportional to $I^2$, and
$s=\sqrt{p^2+q^2}$ is the magnitude of the apparent power injection of the inverter.
Clearly even though  optimal inverter var control can reduce the line losses
and the energy consumption as measured by the CVR term, its deviation from
unity power factor also increases the inverter real power loss.

 \subsection{Overall Problem}

 If we denote the set of the nodes with DC/AC inverters with $I$, then the overall fast time-scale inverter var control problem can be formulated as follows:

\begin{align}
& {\text{min}}
& &\sum_{(i,j)\in E}{r_{ij} |I_{ij}|^2} + \sum_{i}{\alpha_i  |V_i|^2}+\sum_{i\in I} P_{\textrm{loss}}(s_i) \nonumber \\
& \text{s. t.}
&& (\ref{Pbalance})-(\ref{inverter-model}) \nonumber \\
& \text{over}
&&  X := (P, Q, p^g, p^c, q^g, q^c, |V|^2, |I|^2)
\label{Original}
\end{align}
where $|V|^2$ and $|I|^2$ denote the vector variables
$|V|^2 := \{|V_1|^2, \dots, |V_{|N|}|^2\}$ and
$|I|^2 := \{|I_{ij}|^2, \ \forall (i,j) \in E\}$.

It is easy to see that the objective function of the above optimization problem is convex with respect to the state vector $X$. This is because (\ref{line-losses}), (\ref{CVR}) are linear and the inverter loss term (\ref{inverter-loss}) is a linear combination of a norm and a quadratic function. However, the overall problem is nonconvex since its feasible set, i.e., the set of vectors $X$ that satisfy (\ref{Pbalance})--(\ref{inverter-model}), is nonconvex. Therefore the problem is hard to solve in this form.

\subsection{Solution Method}
Note that if we substitute (\ref{update}) into (\ref{Pbalance})--(\ref{Vdrop}), then all constraints become linear with respect to $X$, except the equality constraint (\ref{update}).
This nonlinear equality constraint is the source of
nonconvexity, which we propose to relax:
\begin{equation}
\forall(i,j)\in E: ~~ |I_{ij}|^2 \geq \frac{P_{ij}^2+Q_{ij}^2}{|V_i|^2}
\label{relaxed-update}
\end{equation}
This is equivalent to relaxing the magnitude of currents on all links.
We will later prove these inequalities will be tight in any optimal solution.
Note that relaxed constraints (\ref{relaxed-update}) represent second order cones with respect to $(P,Q, |V|^2, |I|^2)$. In order to use this relaxation to cast problem (\ref{Original}) into a Second Order Cone Program (SOCP), let us introduce the following new variables for every bus and every line, respectively:
\begin{gather}
\nu_i := |V_i|^2 \\
\ell_{ij} := |I_{ij}|^2
\end{gather}
and the following new variables for all $i \in I$, i.e., every bus with an inverter:
\begin{gather}
s_i:=\sqrt{(p^g_i)^2+(q^g_i)^2}\\
t_i:=s_i^2=(p^g_i)^2+(q^g_i)^2
\end{gather}

Now consider the following relaxed SOCP program to solve the fast time-scale inverter var control problem:

\begin{align}
& {\text{min}}
& &\sum_{(i,j)\in E}{r_{ij} \ell_{ij}} + \sum_{i}{\alpha_i  \nu_i}+\sum_{i\in I}{\left(c_v s_i + c_r t_i \right)}
\label{eq:relaxedprob}\\
& \text{s. t.}
&& P_{ij} =\sum_{k:(j,k)\in E} {P_{jk}}+r_{ij} \ell_{ij} +p_j^c-p_j^g \\
&&& Q_{ij} =\sum_{k:(j,k)\in E} {Q_{jk}}+x_{ij} \ell_{ij} +q_j^c-q_j^g - q_j^c \nu_j
\\ 
&&&  \nu_j = \nu_i -2(r_{ij}P_{ij}+x_{ij}Q_{ij}) + (r_{ij}^2+x_{ij}^2) \ell_{ij} \\
&&&  \forall(i,j)\in E: ~~ \ell_{ij} \geq \frac{P_{ij}^2+Q_{ij}^2}{\nu_i} \label{relaxed} \\%
&&&  \forall i\in I: s_i \geq \sqrt{(p^g_i)^2+(q^g_i)^2} \label{norm-relax}\\
&&&  \forall i\in I: t_i \geq (p^g_i)^2+(q^g_i)^2 \label{quad-relax}\\
&&&  \forall i\in I: |q_i^{g}| \leq  \overline{q}_i^g \\
&&& \underline{V}_i^2 \leq \nu_i \leq \overline{V}_i^2 \label{vlimits}\\
&&& \underline{p}_i^c \leq p_i^c ~~,~~ \underline{q}_i^c \leq q_i^c  \label{oversatisfaction}\\
& \text{over}
&&  X := (P, Q, p^g, p^c, q^g, q^c, \nu, \ell, s, t)
\label{eq:vars}
\end{align}

Using relaxations (\ref{norm-relax}), (\ref{quad-relax}) is the typical method of linearizing the Euclidean norms and quadratic terms in the objective function and casting them as second order cone constraints. It is easy to see that these inequalities will be tight in the solution. We have also made two other relaxations in the problem formulation.  First
 the equalities (\ref{update}) in the original problem are relaxed to inequalities
(\ref{relaxed}).  Second, in (\ref{oversatisfaction}), we have used the over-satisfaction of active and reactive loads (see \cite{M11}, \cite{Sojoudi2011}) .
Our key result is

\vspace{2mm}
\begin{theorem}
\label{thm1}
The volt/var control problem (\ref{eq:relaxedprob})--(\ref{eq:vars}) is convex.
Moreover, it is exact, i.e., any optimal solution of (\ref{eq:relaxedprob})--(\ref{eq:vars})
achieves equality in (\ref{relaxed}), (\ref{norm-relax}),(\ref{quad-relax}), and therefore specifies valid and
optimal inverter reactive generation $q_i^g$ and voltage magnitudes $|V_i|$.
\end{theorem}

\begin{proof}
A slightly simpler version of this theorem is proved in \cite{M11}. The only difference is the addition of more linear terms due to inverter losses in the objective function and more second order cone constraints in (\ref{norm-relax}), (\ref{quad-relax}). However, the same proof can be applied here as well since the active and reactive power generation variables are not involved in the argument stated in \cite{M11}.
\end{proof}

\vspace{2mm}
Theorem \ref{thm1}  implies that the original nonconvex optimal inverter control problem (\ref{Original})
can be efficiently solved using the above SOCP relaxation. This opens the way to implement efficient volt/var control in real time to cope with random, rapid, and large fluctuations of solar
generation.

\section{Evaluation}

In this section we evaluate the effectiveness of the proposed inverter control
using data from one of SCE's distribution circuits.
 \begin{figure*}[t]
\centering
\includegraphics[scale=0.45]{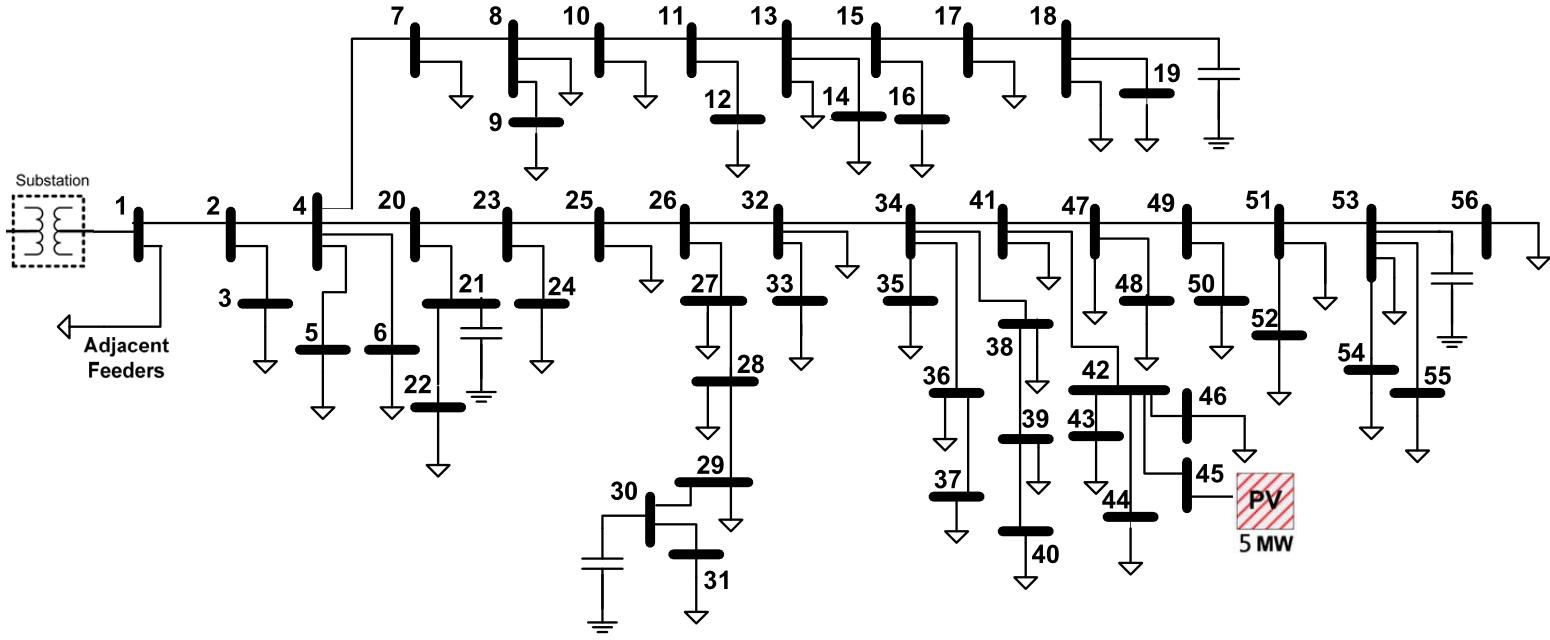}
\caption{Circuit diagram for SCE distribution system.}
 \label{fig:DistrCircuit}
\end{figure*}
The load data in Table \ref{data} are peak values.  Historical data shows a
typical day time loading of around 20\% of the peak with a typical power factor of 0.9.
\begin{centering}
\begin{table*}

\caption{Line impedances, peak spot load kVA, capacitors and PV generation's
nameplate ratings for the distribution circuit in Figure \ref{fig:DistrCircuit}. }
\centering
\scriptsize
\begin{tabular}{|c|c|c|c|c|c|c|c|c|c|c|c|c|c|c|c|c|c|}
\hline
\multicolumn{18}{|c|}{Network Data}\\
\hline
\multicolumn{4}{|c|}{Line Data}& \multicolumn{4}{|c|}{Line Data}& \multicolumn{4}{|c|}{Line Data}& \multicolumn{2}{|c|}{Load Data}& \multicolumn{2}{|c|}{Load Data}&\multicolumn{2}{|c|}{Load Data}\\
\hline
From&To&R&X&From&To& R& X& From& To& R& X& Bus& Peak & Bus& Peak &  Bus &Peak  \\
Bus.&Bus.&$(\Omega)$& $(\Omega)$ & Bus. & Bus. & $(\Omega)$ & $(\Omega)$ & Bus.& Bus.& $(\Omega)$ & $(\Omega)$ & No.&  MVA& No.& MVA& No.& MVA\\
\hline

1	&	2	&	0.160	&	0.388	&	20	&	21	&	0.251	&	0.096	&	39	&	40	&	2.349	&	 0.964	&	3	&	0.057	&	29  &	0.044  & 52& 0.315	 \\
2	&	3	&	0.824	&	0.315	&	21	&	22	&	1.818	&	0.695	&	34	&	41	&	0.115	&	 0.278	&	5	&	0.121	&	31	&	0.053  & 54& 	 0.061	 \\
2	&	4	&	0.144	&	0.349	&	20	&	23	&	0.225	&	0.542	&	41	&	42	&	0.159	&	 0.384	&	6	&	0.049	&	32	&	0.223 & 55&	 0.055	 \\
4	&	5	&	1.026	&	0.421	&	23	&	24	&	0.127	&	0.028	&	42	&	43	&	0.934	&	 0.383	&	7	&	0.053	&	33	&	0.123 & 56&	 0.130	 \\\cline{17-18}
4	&	6	&	0.741	&	0.466   &	23	&	25	&	0.284	&	0.687	&	42	&	44	&	0.506	&	 0.163	&	8	&	0.047	&	34	&	0.067 & \multicolumn{2}{c|}{Shunt Cap}	 \\\cline{17-18}
4	&	7	&	0.528	&	0.468	&	25	&	26	&	0.171	&	0.414	&	42	&	45	&	0.095	&	 0.195	&	9	&	0.068	&	35	&	0.094&   \multicolumn{1}{c|}{Bus} &	 \multicolumn{1}{c|}{Mvar}			 \\\cline{17-18}
7	&	8	&	0.358	&	0.314	&	26	&	27	&	0.414	&	0.386	&	42	&	46	&	1.915	&	 0.769	&	10	&	0.048	&	36	&	0.097&  19& 	 0.6 	 \\
8	&	9	&	2.032	&	0.798	&	27	&	28	&	0.210	&	0.196	&	41	&	47	&	0.157	&	 0.379	&	11	&	0.067	&	37	&	0.281&  21&	0.6 	 \\
8	&	10	&	0.502	&	0.441	&	28	&	29	&	0.395	&	0.369	&	47	&	48	&	1.641	&	 0.670	&	12	&	0.094	&	38	&	0.117&  30&	0.6 		 \\
10	&	11	&	0.372	&	0.327	&	29	&	30	&	0.248	&	0.232	&	47	&	49	&	0.081	&	 0.196	&	14	&	0.057	&	39	&	0.131& 53&	0.6 		 \\\cline{17-18}
11	&	12	&	1.431	&	0.999	&	30	&	31	&	0.279	&	0.260	&	49	&	50	&	1.727	&	 0.709	&	16	&	0.053	&	40	&	0.030& \multicolumn{2}{c|}{Photovoltaic}		 \\\cline{17-18}
11	&	13	&	0.429	&	0.377	&	26	&	32	&	0.205	&	0.495	&	49	&	51	&	0.112	&	 0.270	&	17	&	0.057	&	41	&	0.046& \multicolumn{1}{c|}{Bus} &	 \multicolumn{1}{c|}{Capacity}		 \\\cline{17-18}
13	&	14	&	0.671	&	0.257	&	32	&	33	&	0.263	&	0.073	&	51	&	52	&	0.674	&	 0.275	&	18	&	0.112	&	42	&	0.054&   & \\
13	&	15	&	0.457	&	0.401	&	32	&	34	&	0.071	&	0.171	&	51	&	53	&	0.070	&	 0.170	&	19	&	0.087	&	43	&	0.083&   45 &		 5MW		 \\\cline{17-18}
15	&	16	&	1.008	&	0.385	&	34	&	35	&	0.625	&	0.273	&	53	&	54	&	2.041	&	 0.780	&	22	&	0.063	&	44	&	0.057&  \multicolumn{2}{c|}{} \\
15	&	17	&	0.153	&	0.134	&	34	&	36	&	0.510	&	0.209	&	53	&	55	&	0.813	&	 0.334	&	24	&	0.135	&	46	&	0.134&  \multicolumn{2}{c|}{$V_\textrm{base}$ = 12kV}	 \\
17	&	18	&	0.971	&	0.722	&	36	&	37	&	2.018	&	0.829	&	53	&	56	&	0.141	&	 0.340	&	25	&	0.100	&	47	&	0.045& \multicolumn{2}{c|}{$S_\textrm{base}$ = 1MVA} 	 \\
18	&	19	&	1.885	&	0.721	&	34	&	38	&	1.062	&	0.406	&		&		&		&		 &	27	&	48	&	48	&	0.196&  \multicolumn{2}{c|}{$Z_\textrm{base}= 144 \Omega$ }		 \\
4	&	20	&	0.138	&	0.334	&	38	&	39	&	0.610	&	0.238	&		&		&		&		 &	28	&	38	&	50	&	0.045 &  \multicolumn{2}{c|}{}		 \\

\hline

\end{tabular}
\label{data}
\end{table*}
\end{centering}

\subsection{Simulation setup}

We have chosen one of SCE's distribution feeders with very high penetration of
Photovoltaics.   We use historical SCADA data for load and PV generation to
illustrate the ideas and potential benefits of inverter var control.  This is a very
lightly loaded rural distribution feeder (less than 1MW) in which a 5MW PV has
been integrated almost 6 miles away from the substation.
The circuit diagram of the distribution system is shown Figure \ref{fig:DistrCircuit}
and the various parameters are given in Table \ref{data}.

In the simulation results discussed below, the voltage magnitude of bus
1 at the substation is fixed at 1 pu.  The voltage magnitude bounds at all
other buses are assumed to be 0.97 pu and 1.03 pu.

Figure \ref{fig:Vpcc} shows the voltage magnitude at bus 45, i.e. the Point of Common Coupling (PCC),
as a function of solar output, when the inverter provides no var
control, i.e., unity power factor with $q_i^g =0$, according IEEE 1547.  The load
is in percentage of peak load.  The figure shows that the range of voltage
fluctuation
can exceed 5\% as solar output varies from 0 to 5 MW (its nameplate capacity).
 \begin{figure}[t]
\centering
\includegraphics[scale=0.35]{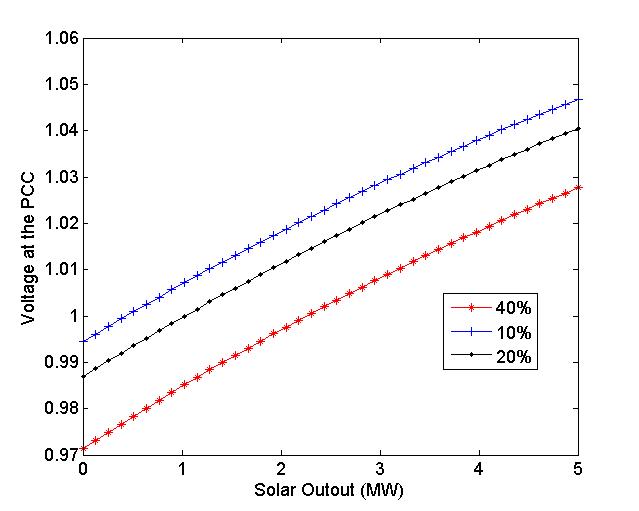}
\caption{Voltage magnitude at the point of common coupling (PCC) vs solar output. }
\label{fig:Vpcc}
\end{figure}
It clearly demonstrates the need for fast timescale inverter var control or an equivalent mitigation strategy
to cope with random, rapid, and large solar output
fluctuations.

\subsection{Optimal inverter var injection}

In this subsection we examine the optimal inverter var injection, defined by $q_i^{g*}$,
at the PV bus under the proposed optimal control, as solar output and load varies.
Multiple factors interact to determine the optimal inverter var injection.
In particular, there is a
tradeoff between the line loss term and the CVR term in the objective
function: a higher voltage magnitude reduces line loss but increases
energy consumption in the CVR term.   The optimal tradeoff is determined
by the solar output and the total load, and the constraints (3\% tolerance)
on the voltage magnitude.   Intuitively, one should
increase var injection (more capacitive) either when the voltage magnitude is
low in order to keep it above its lower bound, or when the solar is high so
as to minimize the line losses in transferring power from the PV bus towards
the substation.   Conversely, one should decrease var
injection (more inductive) either when the voltage magnitude is high in
order to keep it below its upper bound or when the load is high in order
to decrease energy consumption due to the CVR term in the objective
function.
We now take a closer look at these interactions, as solar output and as
load varies

\subsection{Results}

Figures \ref{fig:qvsPV.lowload} and \ref{fig:qvsPV.highload} show the optimal
inverter var injection as a function of solar output for a fixed total load.
At low load (Figure \ref{fig:qvsPV.lowload}),
as solar output increases, the optimal inverter var injection $q_i^{g*}$ initially
increases so as to minimize line losses in transferring power from the PV bus
towards the substation.   Eventually, as the solar output continues to rise
above a threshold the optimal inverter var injection $q_i^{g*}$ decreases
(absorbs var) in order to maintain a voltage magnitude within its upper bound.

The opposite effect dominates at high load (Figure \ref{fig:qvsPV.highload})
when the voltage magnitude is typically well below the upper bound.
 As solar output increases, the
optimal inverter var injection $q_i^{g*}$ decreases so as to reduce the
energy consumption due to the CVR term in the objective function.
We can also see the transition between these two phenomena
in these Figures.

 \begin{figure}[t]
\centering
\includegraphics[scale=0.53]{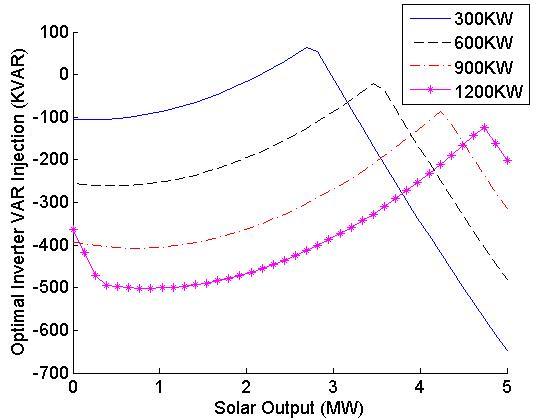}
\caption{Optimal inverter reactive power (in kvar) vs PV output when load is low. }
\label{fig:qvsPV.lowload}
\end{figure}

\begin{figure}[t]
\centering
\includegraphics[scale=0.4]{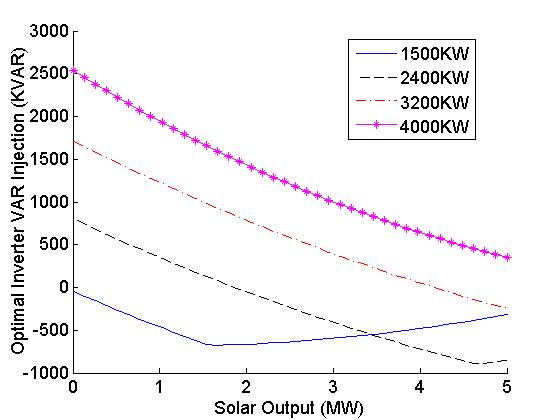}
\caption{Optimal inverter reactive power (in kvar) vs PV output when load is high.}
\label{fig:qvsPV.highload}
\end{figure}

Figure \ref{fig:qvsload} shows the optimal
inverter var injection as a function of total load for a fixed solar output.
\begin{figure}[t]
\centering
\includegraphics[scale=0.4]{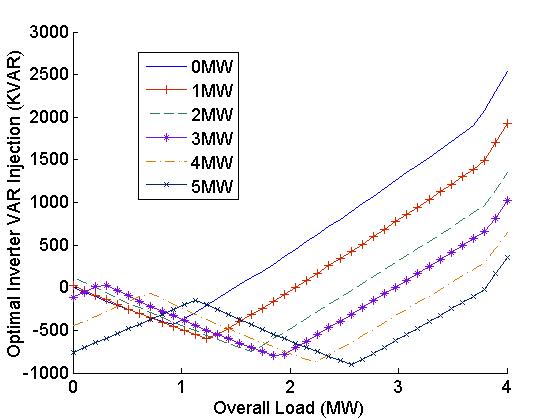}
\caption{Optimal inverter reactive power (in kvar) vs total load.}
\label{fig:qvsload}
\end{figure}
At low solar output, optimal inverter var injection decreases as load
initially increases to as to reduce the energy consumption in the CVR
term, until the load reaches a threshold that reduces the voltages to
near the lower bounds.   Beyond that threshold, the optimal inverter
var injection increases as load increases so as to maintain the voltage
magnitudes above the lower bounds.
At high solar output, on the other hand, the above behavior is preceded
by section where the line losses term dominates over the CVR term.
Then the optimal inverter var injection increases as load
initially increases from zero in order to reduce line losses.

\subsection{Benefits of optimal inverter var control}

We implemented the proposed convex relaxation
of the radial OPF problem (\ref{Pbalance})--(\ref{Vdrop})
and solved it using CVX  optimization toolbox \cite{Boyd} in Matlab.
 In all our simulations,
 we checked the inequality constraint in condition (\ref{relaxed}) for optimal solutions
 of the relaxed problem and confirmed that the inequality
 constraints were all active, i.e., equality holds at the optimal solutions.

We will assess the benefit of the proposed optimal inverter var
control in two ways.  First, when inverters do not participate in var
control (i.e., unity power factor as specified in the current IEEE 1547
standard), the voltage magnitudes may violate the specified limits
when the total load is low and solar power is high.  This represents
an undesirable operation mode.  The proposed optimal inverter
var control should help maintain the voltage magnitudes within
their specified limits and thus enlarge the region of desirable operation
mode.

Second, the net cost is the sum of line losses and energy consumption,
as expressed in (\ref{eq:relaxedprob}), plus the real power loss in the inverter.

We have used typical loss model parameters from \cite{Braun07} for an inverter with maximum efficiency of 97\% to evalue the total cost. Table  \ref{resultsn} summarizes  these two benefits using the distribution
  circuit specified in Table \ref{data}.
\begin{centering}
\begin{table}
\caption{Simulation Results For Some Voltage Tolerance Thresholds }
\centering
    \begin{tabular}{ | l | l | l | }
    \hline
    Voltage Drop & Annual Hours Saved Spending& Average  Power\\
    Tolerance& Outside Feasibility Region &  ~~~Saving  \\ \hline
    ~~~~~~3\%&   ~~~~~~~~~~~~~~~1,214h& ~~~~~1.15\%\\ \hline
    ~~~~~~4\%&  ~~~~~~~~~~~~~~ 223h &~~~~~1.34\%\\ \hline
    ~~~~~~5\%&  ~~~~~~~~~~~~~~~37h&~~~~ 1.42\%\\ \hline
    \end{tabular}
\label{resultsn}
\end{table}
\end{centering}
First, when inverters do not participate in var control,
the feeder spends significant amount of time (e.g., 1,214 hours per
year with 3\% voltage drop tolerance) outside the feasibility region where
voltage magnitudes violate their specified limits.
Under the proposed
optimal inverter  var control, this undesirable operation mode is almost completely
eliminated.  Second the optimal control yields energy savings (above
1 \%), as measured by the total cost that includes the inverter real power loss.
Note that the savings in total cost are calculated only for times where both
the unity power factor control and the optimal control are feasible.  As the
voltage drop tolerance decreases, the unity power factor control becomes
infeasible more often while the optimal control remains feasible, but the
corresponding savings are excluded in the calculation.

%
%
%
%
%
%

\vspace{5mm}

\section{Conclusion}

This study demonstrates the benefits of inverter var control on a fast timescale to mitigate rapid and large voltage fluctuations due to high penetration of photovoltaic generation and the resulting reverse power flow. The problem was formulated as a radial OPF problem that minimizes line losses and energy consumption, subject to constraints on voltage magnitudes. This problem is generally non-convex and hard to solve. A convex relaxation was derived that can be solved efficiently and proved to be exact on radial networks. Finally, we have used it to compute the optimal inverter var injections and illustrate the improvement in both voltage regulation and efficiency under the optimal control. The results are illustrated for an SCE distribution feeder with a very light load and a 5 MW PV system installed 6 miles from the substation.

\vspace{5mm}

\section*{Acknowledgment}
The authors would like to acknowledge Prof.~Mani Chandy of Caltech, the colleagues at Advanced Technology division of Southern California Edison, and the Resnick Institute at Caltech,
for their support and insightful comments.

\vspace{0cm}
\begin{IEEEbiography}[{\includegraphics[width=1in,height=1.25in,clip,keepaspectratio]{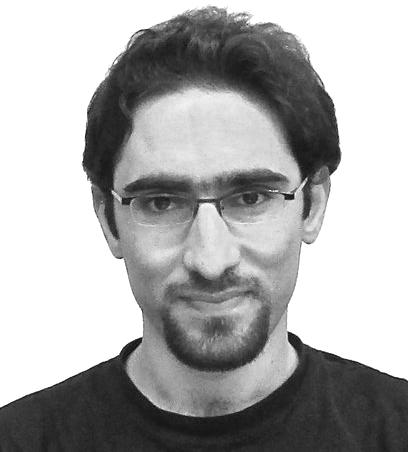}}]{Masoud Farivar}
is a Ph.D. student of electrical engineering at California Institute of Technology, where he is a fellow of the Resnick Institute for Sustainable Energy. Since June of 2010, he has been working in the Advanced Technology division of Southern California Edison on a number of smart grid projects. Masoud received his B.S. degree with a double major in electrical engineering and computer science from Sharif University of Technology, Tehran, Iran, in 2009.

\end{IEEEbiography}

\begin{IEEEbiography}[{\includegraphics[width=1in,height=1.25in,clip,keepaspectratio]{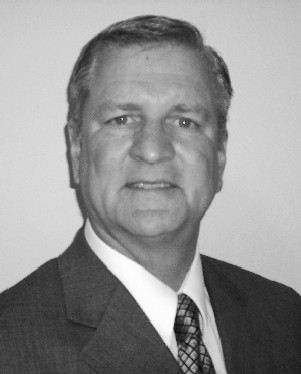}}]{Russell Neal} is strategic program manager for the Southern California Edison Company. He specializes in smart grid issues with an emphasis on distribution systems. Neal is a twenty year IEEE member and graduate of the U. S. Naval Academy. He holds an MEEE from the University of Idaho and an MBA from Azusa Pacific University. He is a registered professional engineer in electrical and nuclear engineering.

\end{IEEEbiography}

\vspace{1cm}
\noindent \small{\textbf{Christopher Clarke }
joined Southern California Edison in April 2004. His primary experience at SCE includes distribution analysis, distribution planning, and power systems modeling. He earned a Master of Science degree in electrical engineering at the University of Southern California and his Bachelor of Science degree in electrical engineering from the University of California at Los Angeles. He is a registered professional engineer in the state of California and a member of the IEEE. }

\vspace{-12cm}
\begin{IEEEbiography}[{\includegraphics[width=1in,height=1.25in,clip,keepaspectratio]{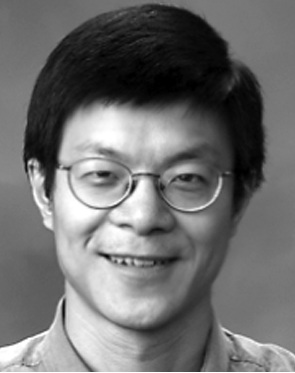}}]{Steven H. Low} (F '08)
received BS from Cornell and PhD from Berkeley, both in EE.
He is a Professor of the Computing \& Mathematical Sciences and Electrical Engineering departments at Caltech, and hold guest faculty positions with the Swinbourne University, Australia and Shanghai Jiaotong University, China. Prior to that, he was with AT\&T Bell Laboratories, Murray Hill, NJ, and the University of Melbourne, Australia. He was on the Editorial Board of the IEEE/ACM Transactions on Networking, IEEE
Transactions on Automatic Control, ACM Computing Surveys, and the
Computer Networks Journal. He is currently a Senior Editor of the
IEEE Journal on Selected Areas in Communications, and series mentor
for its Smart Grid Series.
\end{IEEEbiography}

\end{document}